\documentclass[final,3p,times,12pt]{elsarticle}
\usepackage[english]{babel}
\usepackage{amsthm}
\usepackage{amsmath,amssymb,amsfonts}
\usepackage[dvipsnames,usenames]{color}
\usepackage[utf8]{inputenc}
\usepackage{tikz}
\usetikzlibrary{positioning}
\usepackage{caption}
\usepackage{subcaption}

\newtheorem{thm}{Theorem}[section]
\newtheorem{prop}[thm]{Proposition}
\newtheorem{lem}[thm]{Lemma}
\newtheorem{defi}[thm]{Definition}
\newtheorem{cor}[thm]{Corollary}
\newtheorem{claim}{Claim}

\def\B{\,\square \,}
\def\Z{\mathbb Z}

\journal{European J. Combin.} 

\begin{document}

\begin{frontmatter}

\title{There are no finite partial cubes of girth more than 6 and minimum
degree at least 3}

\author[LJ]{Tilen Marc}
				\ead{tilen.marc@imfm.si}
				
\address[LJ]{Institute of Mathematics, Physics, and Mechanics, Jadranska 19, 1000 Ljubljana, Slovenia}

\begin{abstract}
Partial cubes are graphs isometrically embeddable into hypercubes. We analyze how isometric cycles in  partial cubes behave and derive that every partial cube of girth more than 6 must have vertices of degree less than 3. As a direct corollary we get that every regular partial cube of girth more than 6 is an even cycle. Along the way we prove that every partial cube $G$ with girth more than 6 is a tree-zone graph and therefore $2n(G)-m(G)-i(G)+ce(G)=2$ holds, where $i(G)$ is the isometric dimension of $G$ and $ce(G)$ its convex excess.
\end{abstract}
\end{frontmatter}

\sloppy
\section{Introduction}

Graphs that can be isometrically embedded into hypercubes are called partial cubes. They form a well known class of graphs which inherits many structural properties from hypercubes. For this reason, they were introduced by Graham and Pollack \cite{graham1971addressing} as a model for interconnection networks and latter found different applications, for examples see \cite{BandeltEvolu, eppstein2007media, NadjafiKlavzar}. There has been much theory developed about partial cubes, we direct an interested reader to books  \cite{deza1997geometry,Hammack:2011a} and the survey \cite{ovchinnikov2008partial}. For recent results in the field, see \cite{Albenque2016866,BresarSumenjak,cardinal2014covering,Gologranc14,weizhang}.

Probably the best known subfamily of partial cubes are median graphs \cite{bandelt1984retracts,Hammack:2011a, klmu}. Many questions that are currently open for partial cubes, are long answered for median graphs. Comparing to median graphs, we can learn a lot about partial cubes and even predict certain properties. An example of this is the topic of classifying regular graphs in each class. It was Mulder \cite{mulder1980n} who already in 1980 showed that hypercubes are the only finite regular median graphs; this result has been in some instances generalized also to infinite graphs \cite{bandelt1983infinite,imrich2009transitive,marc2014regular,marc2015vertex}. On the other hand, it seems very difficult to find (non-median) regular partial cubes (particularly in the cubic case),  extensive studies have been made in \cite{bonnington2003cubic, brevsar2004cubic,eppstein2006cubic, klavvzar2007tribes}. In fact, all known cubic partial cubes are planar, besides  the Desargues graph \cite{klavvzar2007tribes}. One of the motivations for this article is to find out why this is so.

One of the most important differences between partial cubes and median graphs is hidden in the cycles of these graphs, particularly in the behavior of isometric and convex cycles. The convex closure of an isomeric cycle in a median graph is a hypercube (for a proof and a generalization to a larger subclass of partial cubes, see \cite{polat2007netlike}). This implies that median graphs that are not trees have girth four, which is far from true in partial cubes. It is an interesting fact that all the known examples of regular  partial cubes have girth four, with the exception of even cycles  and the  middle level graphs (which have girth 6). This motivates the analysis of partial cubes of higher girths.

A motivation for the study of partial cubes with high minimum degree comes from the theory of oriented matroids. Every oriented matroid is characterized by its tope graph, formed by its  maximal covectors \cite{MR1744046}. It is a well known fact that tope graphs are partial cubes, while there is no good characterization of partial cubes that are tope graphs \cite{MR1210100}. It follows from basic properties of oriented matroids that the minimum degree of a tope graph is at least the rank of the oriented matroid it describes. Since the tope graphs of oriented matroids with rank at most 3 are characterized \cite{MR1210100}, there is a special interest in graphs with high minimum degree.

Klavžar and Shpectorov \cite{klavvzar2012convex} proved a certain “Euler-type” formula for partial cubes, concerning convex cycles. Moreover, they defined  the zone graphs of a partial cube: graphs  that emerge if we consider how convex cycles in a partial cube intersect. The latter gave motivation to analyze the space of isometric cycles in partial cubes.

The main contribution of this paper is a theorem which shows that there are no finite partial cubes of girth more than 6 and minimum
degree at least 3. This helps to understand why it is difficult to find regular partial cubes, since it implies that, besides even cycles, there are none with girth more than 6.   To prove the theorem we introduce two concepts - a traverse of isometric cycles and intertwining of isometric cycles - and show some properties of them. We hope that these two definitions will give a new perspective on partial cubes.

In the rest of this section basic definitions and results needed are given. We will consider only simple (possibly infinite) graphs in this paper. The \emph{Cartesian product} $G\, \square \, H$ of graphs $G$ and $H$ is the graph with the vertex set $V(G) \times V(H)$ and the edge set consisting of all pairs $\{(g_1,h_1),(g_2,h_2)\}$ of vertices with $\{g_1,g_2\} \in E(G)$ and $h_1=h_2$, or $g_1=g_2$ and $\{h_1,h_2\} \in E(H)$. \emph{Hypercubes} or \emph{$n$-cubes} are Cartesian products of $n$-copies of $K_2$. We say a subgraph $H$ of $G$ is \emph{convex} if for every pair of vertices in $H$ also every shortest path connecting them is in $H$. On the other hand, a subgraph is \emph{isometric} if for every pair of vertices in $H$ also some shortest path connecting them is in $H$. A \emph{partial cube} is a graph that is isomorphic to an isometric subgraph of some hypercube.

For a graph $G$, we define the relation $\Theta$ on the edges of $G$ as follows: $ab \Theta xy$ if $d(a,x) + d(b,y)\neq d(a,y) + d(b,x)$, where $d$ is the shortest path distance function. In partial cubes $\Theta$ is an equivalence relation (in fact a bipartite graph is a partial cube if and only if $\Theta$ is an equivalence relation \cite{winkler1984isometric}), and we write $F_{uv}$ for the set of all edges that are in relation $\Theta$ with $uv$. We define $W_{uv}$ as the subgraph induced by all vertices that are closer to vertex $u$ than to $v$, that is $W_{uv}=\langle \{ w: \   d(u,w) < d(v, w) \} \rangle $. In a partial cube $G$, subgraphs $W_{uv}$ are convex, and the sets $V(W_{uv})$ and $V(W_{vu})$ partition $V(G)$, with  $F_{uv}$ being the set of edges joining them. We define $U_{uv}$ to be the subgraph induced by the set of vertices in $W_{uv}$ which have a neighbor in $W_{vu}$. For  details and further results, see \cite{Hammack:2011a}.

We shall need a few simple results about partial cubes. It $u_1v_1 \Theta u_2v_2$ with $u_2 \in U_{u_1v_1}$, then $d(u_1,u_2)=d(v_1,v_2)=d(u_1,v_2)-1=d(u_2,v_1)-1$. A path $P$ is  a shortest path or a \emph{geodesic} if and only if it has all of its edges in pairwise different $\Theta$ classes. For fixed $u,v$ all shortest $u,v$-paths pass the same $\Theta$-classes of $G$. If $C$ is a cycle and $e$ an edge on $C$, then there is another edge on $C$ in relation $\Theta$ with $e$. We denote with $I(a,b)$ the \emph{interval} from vertex $a$ to vertex $b$, i.e.~the induced subgraph on all the vertices that lie on some shortest $a,b$-path. In a partial cube, for every vertices $a$ and $b$, the subgraph $I(a,b)$ is convex.  For the details, we again refer to \cite{Hammack:2011a}.

In \cite{klavvzar2012convex}, the following definition was given: Let $G$ be a partial cube and $F$ be some equivalence class of relation $\Theta$. The \emph{$F$-zone graph}, denoted with $Z_F$, is the graph with $V(Z_F) = F$, vertices $f$ and $f'$ being adjacent in $Z_F$ if they belong to a common convex cycle of $G$. We call a partial cube whose all zone graphs are trees a \emph{tree-zone} partial cube.

For a graph $G$, we shall denote with $g(G)$ the girth of $G$, i.e.~the length of a shortest cycle in $G$. In this paper we will consider, beside finite, also infinite, locally finite (every vertex has at most finitely many neighbors) graphs. For such graphs the following definition makes sense: If $d\in \mathbb{N}$, let $B_d(v)$ be the number of vertices at distance at most $d$ from a vertex $v$ of a graph $G$. If $B_d(v)$ is bounded from below by some exponential function in $d$, we say that $G$ has an \emph{exponential growth}. The definition is independent of the choice of the vertex in $G$.

\section{Results}

We start with a definition that we will use throughout the rest of the paper.

\begin{defi}
Let $v_1u_1\Theta v_2u_2$ in a partial cube $G$, with $v_2 \in U_{v_1u_1}$. Let $C_1,\ldots, C_n$, $n\geq 1$, be a sequence of isometric cycles such that $v_1u_1$ lies only on $C_1$, $v_2u_2$ lies only on $C_n$, and each pair $C_i$ and $C_{i+1}$, for $i\in \{1,\ldots,n-1\}$, intersects in exactly one edge and this edge is in $F_{v_1u_1}$, all the other pairs do not intersect. If the shortest path from $v_1$ to $v_2$ on the union of $C_1,\ldots, C_n$ is a shortest $v_1, v_2$-path in $G$, then we call $C_1,\ldots, C_n$ a \emph{traverse} from $v_1u_1$ to $v_2u_2$.
\end{defi}

Every isometric cycle in a partial cube has its antipodal edges in relation $\Theta$. Using this fact, we see that if $C_1,\ldots, C_n$ is a traverse from $v_1u_1$ to $v_2u_2$, then also the shortest path from $u_1$ to $u_2$ on the union of $C_1,\ldots, C_n$ is isometric in $G$, since it must have all its edges in different $\Theta$-classes. We will call this $u_1,u_2$-shortest path the \emph{$u_1,u_2$-side of the traverse}  and, similarly, the shortest $v_1,v_2$-path on the union of $C_1,\ldots, C_n$ the \emph{$v_1,v_2$-side of the traverse}. The  length of these two shortest paths is the \emph{length of the traverse}. If all isometric cycles on a traverse $T$ are convex cycles, we will call $T$ a \emph{convex traverse}.

The next lemma is inspired by results from \cite{klavvzar2012convex}.

\begin{lem}\label{lem:cycles}
Let $v_1u_1\Theta v_2u_2$ in a partial cube $G$. Then there exists a convex traverse from $v_1u_1$ to $v_2u_2$.
\end{lem}

\begin{proof}
Assume that this is not the case and let $v_1u_1$ and $v_2u_2$ be counterexample edges with distance between them as small as possible. Since $G$ is connected, there is a shortest $u_1,u_2$-path $R_1$ and a shortest $v_1,v_2$-path $R_2$. We claim that the cycle $C$ on $u_1v_1$, $R_2$, $v_2u_2$, and  $R_1$ is convex.

First we will prove that there is no path  connecting vertices $r_1\in V(R_1)$ and $r_2\in V(R_2)$ that is incident with $C$ only in its endpoints and is shorter or of the same length  as a shortest $r_1,r_2$-path on $C$. For the sake of contradiction, assume that such a path $S$ exists.

Since $F_{v_1u_1}$ is a cut, there is an edge $v'u'$ on $S$ that is in $F_{v_1u_1}$. Moreover, since $ I(v_1,u_2)$ is a convex subgraph and $r_1, r_2 \in I(v_1,u_2)$, the edge $v'u'$ lies on some shortest $v_1,u_2$-path. Thus it holds $d(v_1,v_2)+1=d(v_1,u_2)=d(v_1,v')+1+d(u',u_2)$ which implies $d(v_1,v_2)=d(v_1,v')+d(v',v_2)$.

Distance between $v_1u_1$ and $v'u'$ is smaller than distance between $v_1u_1$ and  $v_2u_2$.
Therefore there exists a convex traverse from $v_1u_1$ to $v'u'$ and, similarly, a convex traverse from $v'u'$ to $v_2u_2$. We argue that the union of both traverses is a convex traverse from $v_1u_1$ to $v_2u_2$. The shortest $v_1,v_2$-path on the union of traverses is of length at most $d(v_1,v')+d(v',v_2)$, which by the above paragraph equals $d(v_1,v_2)$. The latter implies that all the vertices on the traverse from $v'u'$ to $v_2u_2$ are at distance at least $d(v_1,v')$ from $v_1u_1$, while vertices on the  traverse from $v_1u_1$ to $v'u'$ are at distance at most $d(v_1,v')$ from $v_1u_1$. The only vertices on both traverses that are at distance $d(v_1,v')$ from $v_1u_1$ are $v'$ and $u'$. Thus the convex cycles on both traverses have the right intersections to form a convex traverse. A contradiction with the assumption that a convex traverse does not exist.

Now we can prove that $C$ is convex. We have already proved that every pair $r_1\in V(R_1)$ and $r_2\in V(R_2)$ is connected on $C$ with a shortest path. Since every pair of vertices on $R_1$ or a pair on $R_2$ is connected by a shortest path by definition, the cycle $C$ must be isometric. Assume that there is a path $S$ connecting $r_1,r_2\in V(C)$ that has the same length as a shortest path on $C$ and has only its endpoints on $C$. We have proved that $S$ cannot have its endpoints on $R_1$ and $R_2$, thus $r_1$ and $r_2$ must be both in $R_1$ or both in $R_2$. Without loss of generality, assume that they are in $R_1$. Then there exists a shortest $u_1,u_2$-path $R_1'$, different from $R_1$. Now the same arguments that prove that $C$ is isometric also prove that the cycle $C'$ on $u_1v_1$, $R_2$, $v_2u_2$, and  $R_1'$ is isometric. Isometric cycles $C$ and $C'$ cannot simultaneously exist since antipodal edges in an isometric cycle are in relation $\Theta$, while no vertex can be incident with two edges in the same $\Theta$-class.

We have proved that $C$ is convex. This is a contradiction with the assumption of the existence of edges without a convex traverse.

%
%
\end{proof} 

The next lemma turns out to be extremely useful when working with isometric cycles in a partial cube.

\begin{lem}\label{lem:pastecycle}
Let $P=u_0u_1\ldots u_m$ be a geodesic in a partial cube. If there is some other shortest $u_0,u_m$-path, then there exists a convex cycle  of the form $(u_{i}u_{i+1}\ldots u_{j}w_{j-1}w_{j-2}\ldots w_{i+1} )$ for some $0\leq i<j-1 \leq m-1$ and $j-i-1$ vertices $w_{i+1},\ldots ,w_{j-1}$ not on $P$.
\end{lem}

\begin{proof}

Assume that this is not the case and let $P=u_0u_1\ldots u_m$ and $P'$ be two different $u_0,u_m$-geodesics for which the lemma does not apply. Without loss of generality, assume that the length of $P$ is minimal among all counterexamples of the lemma.

By the minimality assumption, the paths $P$ and $P'$ intersect only in $u_0$ and $u_m$. Denote the vertices of $P'$ with $u_0z_1z_2\ldots z_{m-1}u_m$ and let $C$ be the cycle formed by $P$ and $P'$.

Beside $u_{0}z_{1}$ itself, there must an additional edge on $C'$, that is in relation $\Theta$ with $u_{0}z_{1}$. Since $P'$ is a geodesic, this edge is on $P$.
Let $u_{k-1}u_{k}\in F_{u_{0}z_{1}}$, for some $0<k\leq m$. By Lemma \ref{lem:cycles}, there is a convex traverse from $u_{0}z_{1}$ to $u_{k-1}u_{k}$. First,  assume that the  path $P''=u_{0}u_{1}\ldots u_{k-1}$ is the $u_{0},u_{k-1}$-side of it. Then the last convex cycle on this traverse is of the form $(u_{k'} u_{k'+1}\ldots u_{k-1} u_{k}w_{k-1}\ldots w_{k'+1})$ for some $0\leq k'\leq k-2$ and some vertices $w_{k'+1},\ldots ,w_{k-1}$ not on $P$ (they do not lie on $P$ since the cycle is convex). We have found the desired cycle.

On the other hand, assume that $P''$ is not the $u_{0},u_{k-1}$-side of a traverse from $u_{0}z_{1}$ to $u_{k-1}u_{k}$. The geodesic $P''$ is shorter than $P$, and there exists another shortest $u_{0},u_{k-1}$-path, namely the $u_{0},u_{k-1}$-side of a traverse from $u_{0}z_{1}$ to $u_{k-1}u_{k}$. By the minimality assumption, there exists a convex cycle of the form $(u_{i}u_{i+1}\ldots u_{j} w_{j-1}w_{j-2}\ldots w_{i+1})$ for some $0\leq i<j-1 \leq k-2$ and some vertices $w_{i+1},\ldots ,w_{j-1}$ not on $P''$. Since $P''$ is a subpath of $P$, we have again found a cycle from the assertion of the lemma.
\end{proof}

\begin{cor}\label{cor:2pos}
Let $u_0v_0 \Theta u_mv_m$ hold in a partial cube. If $P=u_0u_1\ldots u_m$ is a geodesic, then $P$ is the $u_0,u_m$-side of a convex traverse from $v_0u_0$ to  $v_mu_m$ or there is a convex cycle of the form $(u_{i}w_{i+1}\ldots w_{j-1} u_{j}u_{j-1}\ldots u_{i+1})$ for some $0\leq i<j-1 \leq m-1$ and $j-i-1$ vertices $w_{i+1},\ldots ,w_{j-1}$ not on $P$.
\end{cor}

\begin{proof}
If $P$ is not the $u_0,u_m$-side of a convex traverse from $v_0u_0$ to $v_mu_m$, then we have another shortest $u_0,u_m$-path different from $P$, namely the $u_0,u_m$-side of a convex traverse provided by Lemma \ref{lem:cycles}. Thus, by the previous lemma, the corollary follows.
\end{proof}

In the following we will work with isometric cycles that intersect pairwise in more than a vertex or an edge. We will be particularly interested in the following type of intersections.

\begin{defi}
Let $C_1=(v_0v_1\ldots v_mv_{m+1}\ldots v_{2m+2n_1-1})$ and $C_2=(u_0u_1\ldots u_mu_{m+1}\ldots u_{2m+2n_2-1})$ be isometric cycles with $u_0=v_0,\ldots, u_m=v_m$ for $m\geq 2$, and all the other vertices pairwise different. Then we say that $C_1$ and $C_2$ \emph{intertwine} and define $i(C_1,C_2)=n_1+n_2\geq 0$ as the \emph{residue of intertwining}.
\end{defi}

 Notice that we can calculate the residue of intertwining as $i(C_1,C_2)=(l_1+l_2-4m)/2$, where $l_1$ is the length of $C_1$, $l_2$ the length of $C_2$, and $m$ the number of edges in the intersection. Also notice that in a partial cube, $m$ can be at most half of $l_1$ or $l_2$. Let us prove the latter:  If $m>l_1/2$, then the fact that antipodal edges in an isometric cycle are in relation $\Theta$ implies that $C_1$ is determined by the intersection. Moreover, the path in the intersection is not isometric, thus it must cover more than half of $C_2$, i.e.~$m>l_2/2$. Thus also $C_2$ is determined by the intersection, and consequently we have $C_1=C_2$.
 
%
%

\begin{lem}\label{lem:intersecttointertwine}
Let $G$ be a partial cube and let two isometric cycles intersect in at least two non-adjacent vertices. Then there exist two isometric cycles that intertwine.
\end{lem}

\begin{proof}
Assume that we have two isometric cycles $C_1$ and $C_2$ that intersect in at least two non-adjacent vertices, say $v_1$ and $v_m$. If they do not intertwine, we can assume that that a shortest $v_1,v_m$-path $P_1$ on $C_1$ intersects with a shortest $v_1,v_m$-path $P_2$ on $C_2$ only in the endpoints. Denote the vertices on $P_1$ with $v_1v_2\ldots v_m$, and the vertices on $P_2$ with $v_1u_2u_3\ldots u_{m-1}v_m$. We analyze two cases: First, assume that the length of $P_1$ is strictly less than half of the length of $C_1$. The path $P_1$ is isometric, and since $P_1\neq P_2$, by Lemma \ref{lem:pastecycle}, we have an isometric cycle $C$ of the form $C=(v_{k}w_{k+1}\ldots w_{l-1} v_{l}v_{l-1}\ldots v_{k+1})$ for some $1\leq k<l-1 \leq m-1$ and some vertices $w_{k+1},\ldots ,w_{l-1}$ not on $P_1$. 
Notice that vertices $w_{k+1},\ldots ,w_{l-1}$ do not intersect vertices on $C_1$. To see the latter, let $v_0$ be the neighbor of $v_1$ on $C_1$ different from $v_2$, and $v_{m+1}$ be the neighbor of $v_m$ on $C_1$ different from $v_{m-1}$. Since the length of $C_1$ is strictly greater than the length of $C$, no edge on $P_1$ is in relation $\Theta$ with $v_0v_1$ and $v_mv_{m+1}$. This implies that $V(C)\subset W_{v_1v_0}\cap W_{v_mv_{m+1}}$, while vertices of $C_1$ that do not lie on $P_1$ are  in $W_{v_0v_1}\cup W_{v_{m+1}v_{m}}$. We have proved that  $C_1$ and $C$ intertwine.

Second, assume that the length of $P_1$ is exactly half of the length of $C_1$. There is exactly one edge on $P_1$ that is in relation $\Theta$ with the edge $v_1u_2$ (first edge of $P_2$), say $v_{i-1}v_{i}\Theta v_1u_2$, for some $1< i\leq m$. 
Since the length of $P_1$ is exactly half of the length of $C_1$, the edge $v_mv_{m-1}$ is in relation $\Theta$ with edge $v_0v_1$, where $v_0$ is again the neighbor of $v_1$ on $C$ different from $v_2$. Then $v_mv_{m-1}$ is not in relation $\Theta$ with $v_1u_2$, since no two incident edges can be in relation $\Theta$. Thus $i<m$. Let $P$ be a shortest $v_i,u_2$-path in $W_{u_2v_1}$. The extension of $P$ with an edge $u_2v_1$ is a shortest $v_1,v_i$-path of length less than half of the length of $C_1$ and different from $v_1v_2\ldots v_i$. As before we use Lemma \ref{lem:pastecycle} to obtain an isometric cycle $C$, $C\neq C_1$, of the form $C=(v_{k}w_{k+1}\ldots w_{l-1} v_{l}v_{l-1}\ldots v_{k+1})$, for some $0\leq k<l-1 \leq m-1$ and some vertices $w_{k+1},\ldots ,w_{l-1}$ not on $P_1$. For the same reasons as before, all vertices $w_{k+1},\ldots ,w_{l-1}$ are disjoint with vertices of $C_1$, thus $C_1$ and $C$ intertwine.
\end{proof}

%

For the next result, let $X$ be the graph from Figure 1.

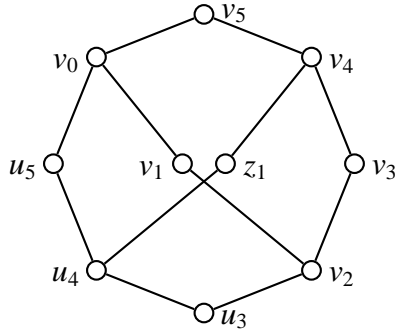
\begin{figure}[h]

\centering
\begin{tikzpicture}[style=thick,scale=2,label distance=0pt,node distance=10pt,rotate=-45]
\tikzstyle{vertex}=[draw, circle, inner sep=0pt, minimum size =7pt]
\node (0) at (0.3  ,0.3) [vertex,label={[shift={(-0.4,-0.5)}]$u_5$}] {};
\node (1) at ( 1  ,0) [vertex,label={[shift={(-0.4,-0.5)}]$u_{4}$ }] {};
\node (2) at (1.7  ,0.3) [vertex,label={[shift={(0.4,-0.5)}]$u_3$}]{};
\node (3) at ( 0  ,1) [vertex,label={[shift={(-0.4,-0.5)}]$v_0$}] {};
\node (4) at (0.9,0.9  ) [vertex,label={[shift={(-0.4,-0.5)}]$v_{1}$ }] {};
\node (5) at ( 2,1) [vertex,label={[shift={(0.4,-0.5)}]$v_{2}$ }] {};
\node (6) at (0.3,1.7) [vertex,label={[shift={(0.4,-0.4)}]$v_{5}$ }] {};
\node (7) at ( 1,2) [vertex,label={[shift={(0.4,-0.5)}]$v_4$}] {};
\node (8) at (1.7  , 1.7) [vertex,label={[shift={(0.4,-0.5)}]$v_{3}$ }] {};
\node (9) at ( 1.1  , 1.1) [vertex,label={[shift={(0.4,-0.5)}]$z_{1}$ }] {};

\draw (0) -- (1);
\draw (1) -- (2);
\draw (0) -- (3);
\draw (3) -- (4);
\draw (4) -- (5);

\draw (2) -- (5);
\draw (3) -- (6);
\draw (5) -- (8);
\draw (6) -- (7);

\draw (7) -- (8);

\draw (1) -- (9);
\draw (9) -- (7);

\end{tikzpicture}
\caption{Graph $X$ with labels}
\end{figure}

\begin{prop}\label{prop:girth2}
If $G$ is a partial cube with $g(G)>6$, then every pair of isometric cycles in $G$ meets in either exactly one edge, or exactly one vertex, or not at all. Moreover, the same holds if $g(G)=6$, provided that $G$ contains no isometric subgraph isomorphic to $X$.
\end{prop}

\begin{proof}
Le $G$ be a partial cube with $g(G)\geq 6$.
Assume that the proposition does not hold and let $(v_0v_1\ldots v_mv_{m+1}\ldots v_{2m+2n_1-1})$ and $(u_0u_1\ldots u_mu_{m+1}\ldots u_{2m+2n_2-1})$ be intertwining isometric cycles with $u_0=v_0,\ldots, u_m=v_m$ for $m\geq 2$ and with their intertwining residue  $n_1+n_2$  minimal among all intertwining isometric cycles in $G$. Such cycles exist by Lemma \ref{lem:intersecttointertwine}. To prove the assertion of the proposition, we shall show that $G$ contains an isometric subgraph isomorphic to $X$. First we prove the following claim.


\begin{claim}\label{claim:1}
Let $T$ be a traverse in $G$ of length $n\leq n_1+n_2$ with a side $P_1=z_0z_1 \ldots z_{n}$. Then there is no isometric cycle $C_1$ of the form $C_{1} = (z_kz_{k+1}\ldots z_{l}w_{l-1}w_{l-2}\ldots w_{k+1})$ for some $0\leq k<l-1\leq n-1$ and vertices $w_{l-1},w_{l-2},\ldots, w_{k+1}$ not on $P_1$.
\end{claim}

\begin{proof}



Assume that $T$, $P_1$ and $C_1$ from the statement exist. Let $P_2=s_0s_1\ldots s_n$ be the other side of $T$. Notice that in this notation $T$ is a traverse from $z_0s_0$ to $z_ns_n$. The length of $C_{1}$ is $2(l-k)$, therefore $l-k \geq 3$, since $g(G)\geq 6$. Consider the edge $z_kz_{k+1}$ on $P_1$. Since it is on the side of $T$, it lies in some isometric cycle $D_1$ on $T$. 

First assume that also the edge $z_{k+1}z_{k+2}$ lies on $D_1$, and let $z_{i'}z_{i'+1}\ldots z_{j'}$ be the vertices of $D_1$ on $P_1$ (i.e.~in $W_{z_0s_0}$), and $s_{i'}s_{i'+1}\ldots s_{j'}$ be vertices of $D_1$ on $P_2$ (i.e.~in $W_{s_0z_0}$). Cycles $C_1$ and $D_1$ intertwine: they intersect in the common consecutive vertices on $P_1$ (at least in $z_kz_{k+1}z_{k+2}$ by our assumption and the fact that $l-k \geq 3$), while all the other vertices are disjoint since $s_{i'}, \ldots s_{j'}\in W_{s_0z_0}$ and $C_1$ lies in  $W_{z_0s_0}$.

On the other hand, if $z_{k+1}z_{k+2}$ lies on an isometric cycle $D_2$ on $T$, different from $D_1$, then also $z_{k+2}z_{k+3}$ must lie on $D_2$. If not, then $D_2$ is of the form $(z_{k+1}z_{k+2}s_{k+2}s_{k+1})$ for some $s_{k+1}s_{k+2} \in W_{s_0z_0}$, thus $g(G)=4$, a contradiction. We see that $C_1$ and $D_2$ intertwine: they intersect in the common consecutive vertices on $P_1$ (at least in $z_{k+1}z_{k+2}z_{k+3}$ by our assumption and the fact that $l-k \geq 3$) while all the other vertices are different.

To sum up both cases, there exists an isometric cycle $C_{2}=(z_{i}z_{i+1}\ldots z_{j} s_{j}\ldots s_i)$ on $T$ ($C_2$ is in fact $D_1$ or $D_2$), for some $0\leq i< j-1\leq n-1$ and some vertices $ s_{j},\ldots, s_i\in V(P_2)$, that intertwine with $C_1$.
We have multiple options for the positions of $C_1$ and $C_2$, i.e.~whether $k\leq i < l \leq j$, $k\leq i < j \leq l$, $i\leq k < l \leq j$, or $i\leq k < j \leq l$.

If $k\leq i < l \leq j$, then the assumption that the cycles meet in at least two edges gives $l-i\geq 2$. It also holds $j-k\leq n$, and $n\leq n_1+n_2$ by the choice of $T$.
Now we calculate the residue of intertwining of $C_1$ and $C_2$. It can be calculated by the formula $(l_1+l_2-4l_3)/2$ where $l_1$ is the length of $C_1$, $l_2$ the length of $C_2$, and $l_3$ the length of the intersection. We have $l_1=2(l-k)$, $l_2=2(j-i)+2$, while the length of intersection is $l-i$. Thus it holds:
$$i(C_1,C_2)= (2(l-k)+2(j-i)+2- 4(l-i))/2= 1+(j-k)-(l-i)\leq 1+n_1+n_2- 2<n_1+n_2.$$
The one but last inequality holds by the inequalities in the previous paragraph.
This is a contradiction with the minimality assumption. 

In the case $k\leq i < j \leq l$ (for the same reasons as in the previous case) we have $l-k\leq n_1+n_2$, $j-i\geq 2$, and the length of the intersection is $j-i$. Thus:
$$i(C_1,C_2)= (2(l-k)+2(j-i)+2- 4(j-i))/2= (l-k)-(j-i)+1\leq n_1+n_2- 2+1<n_1+n_2.$$

The two remaining cases are similar. If $i\leq k < l \leq j$,  we have $j-i\leq n_1+n_2$, $l-k\geq 2$, and the length of the intersection is $l-k$. Thus:
$$i(C_1,C_2)= (2(l-k)+2(j-i)+2- 4(l-k))/2= -(l-k)+(j-i)+1\leq -2+n_1+n_2+1<n_1+n_2.$$

Finally, if $i\leq k < j \leq l$  we have $l-i\leq n_1+n_2$, $j-k\geq 2$, and the length of the intersection is $j-k$. Thus:
$$i(C_1,C_2)= (2(l-k)+2(j-i)+2- 4(k-j))/2= (l-i)-(j-k)+1\leq n_1+n_2- 2+1<n_1+n_2.$$

We have obtained a contradiction in all the cases, which proves the claim.
\end{proof}

We now analyze the relation $\Theta$ in cycles $(v_0v_1\ldots v_mv_{m+1}\ldots v_{2m+2n_1-1})$ and $(u_0u_1\ldots u_mu_{m+1}\ldots u_{2m+2n_2-1})$.
We have $v_{2m+n_1}v_{2m+n_1-1} \Theta v_{m}v_{m-1}$ and $u_{2m+n_2}u_{2m+n_2-1} \Theta u_{m}u_{m-1}$. Therefore $v_{2m+n_1}v_{2m+n_1-1}$ is in relation $\Theta$ with  $u_{2m+n_2}u_{2m+n_2-1}$. Similarly, $v_{2m+n_1-1}v_{2m+n_1-2}$ is in relation $\Theta$ with $u_{2m+n_2-1}u_{2m+n_2-2}$.

\begin{figure}[h]

\centering
\begin{tikzpicture}[style=thick,scale=1,label distance=0pt,node distance=0pt]
\tikzstyle{vertex}=[draw, circle, inner sep=0pt, minimum size =7pt]

\node (0) at (0  ,0) [vertex,label={[shift={(-0.0,-0.8)}]$v_0$}] {};
\node (1) at ( 1  ,0) [vertex,label={[shift={(-0.0,-0.8)}]$v_{1}$ }] {};
\node (2) at (2  ,0) [vertex,label={[shift={(-0.0,-0.8)}]$v_{2}$ }] {};
\node (3) at ( 3  ,0) [vertex,label=below left:{ }] {};
\node (4) at (5, 0 ) [vertex,label={[shift={(-0.0,-0.8)}]$v_{m}$ }] {};
\node (5) at ( 5.7,0.7) [vertex,label=below left:{}] {};
\node (6) at (5.7,2.3) [vertex,label=below left:{ }] {};
\node (7) at ( 5,3) [vertex,label={[shift={(0.6,-0.3)}]$v_{m+n_1}$ }] {};
\node (8) at (3  , 3) [vertex,label=above:{ }] {};
\node (9) at ( 2  , 3) [vertex,label={[shift={(0.4,-0.1)}]$v_{2m+n_1-2}$ }] {};
\node (10) at ( 1  , 3) [vertex,label={[shift={(0.05,-0.1)}]$v_{2m+n_1-1}$ }] {};
\node (11) at ( 0  , 3) [vertex,label={[shift={(-0.6,-0.3)}]$v_{2m+n_1}$ }] {};
\node (12) at ( -0.7  , 2.3) [vertex,label=left:{ }] {};
\node (13) at ( -0.7  , 0.7) [vertex,label=left:{ }] {};

\node (14) at ( 5.7,-0.7) [vertex,label=below left:{}] {};
\node (15) at (5.7,-2.3) [vertex,label=below left:{}] {};
\node (16) at ( 5,-3) [vertex,label={[shift={(0.6,-0.7)}]$u_{m+n_2}$ }] {};
\node (17) at (3  , -3) [vertex,label=below:{}] {};
\node (18) at ( 2  , -3) [vertex,label={[shift={(0.4,-0.8)}]$u_{2m+n_2-2}$ }] {};
\node (19) at ( 1  , -3) [vertex,label={[shift={(0.05,-0.8)}]$u_{2m+n_2-1}$ }] {};
\node (20) at ( 0  , -3) [vertex,label={[shift={(-0.5,-0.7)}]$u_{2m+n_2}$ }] {};
\node (21) at ( -0.7  , -2.3) [vertex,label=left:{ }] {};
\node (22) at ( -0.7  , -0.7) [vertex,label=left:{ }] {};

\node (23) at (1  , 1.5) {$P_1$};
\node (24) at ( 1.4  , -2.5) [label=below:{ }] {};
\node (25) at ( 2.7  , 1.5) {$P_2$};
\node (26) at ( 2.4  , -2.5) [label=left:{}] {};

\draw[dashed] (10) to[bend left=15] (19);

\draw[dashed] (1.38,1.9) --(2.2,1.9);
\draw[dashed] (1.58,0.8) --(2.4,0.8);

\draw[dashed] (1.38,-1.9) --(2.2,-1.9);
\draw[dashed] (1.58,-0.8) --(2.4,-0.8);

\draw[dashed] (18) to[bend right=15] (9);

\draw (4) -- (14);
\draw (15) -- (14);
\draw (15) -- (16);
\draw[dashed] (17) -- (16);
\draw (17) -- (18);
\draw (19) -- (18);
\draw (20) -- (21);
\draw (20) -- (19);
\draw[dashed] (21) -- (22);
\draw (22) -- (0);

\draw (0) -- (1);
\draw (1) -- (2);
\draw (2) -- (3);
\draw[dashed] (3) -- (4);
\draw (4) -- (5);
\draw[dashed] (5) -- (6);
\draw (7) -- (6);
\draw[dashed] (8) -- (7);
\draw (8) -- (9);
\draw (10) -- (9);
\draw (10) -- (11);
\draw (11) -- (12);
\draw[dashed] (12) -- (13);
\draw (13) -- (0);

\end{tikzpicture}
\caption{A situation from the proof of Proposition \ref{prop:girth2}}
\end{figure}
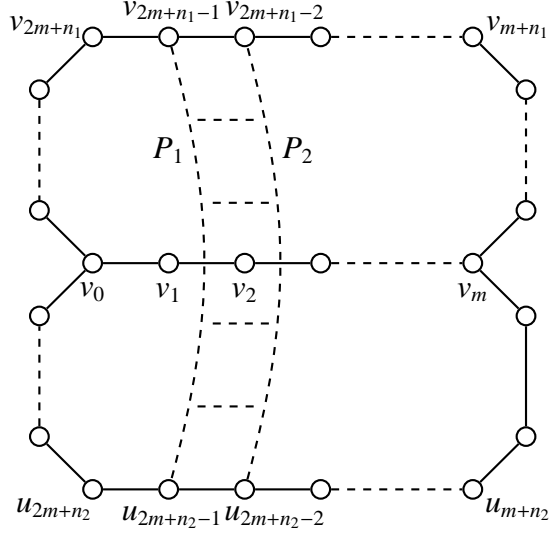

By Lemma \ref{lem:cycles}, there is a traverse $T$ from $v_{2m+n_1-1}v_{2m+n_1-2}$ to $u_{2m+n_2-1}u_{2m+n_2-2}$. Denote with $P_1$ the  $v_{2m+n_1-1},u_{2m+n_2-1}$-side of $T$ and with $P_2$ the $v_{2m+n_1-2},u_{2m+n_2-2}$-side of $T$ (see Figure 2).  
Moreover, let the vertices on $P_1$ be denoted by $z_0,z_1,\ldots, z_{n}$, where $z_0=v_{2m+n_1-1}$, $z_{n}=u_{2m+n_2-1}$, and $n$ is the length of $T$.

We first show that the length $n$ is at most $n_1+n_2$. By definition $n=d(v_{2m+n_1-1},u_{2m+n_2-1})$, and $d(v_{2m+n_1-1},u_{2m+n_2-1})=d(v_{2m+n_1},u_{2m+n_2})$ since  $v_{2m+n_1-1}v_{2m+n_1} \Theta u_{2m+n_2-1}u_{2m+n_2}$. On the other hand, there exists a $v_{2m+n_1},u_{2m+n_2}$-path of length $n_1+n_2$, namely the path $v_{2m+n_1}v_{2m+n_1+1}\ldots v_0u_{2m+2n_2-1} \ldots u_{2m+n_2}$. This proves the inequality.

Considering $P_1$ and edges $v_{2m+n_1}v_{2m+n_1-1}$ and $u_{2m+n_2}u_{2m+n_2-1}$, then, by Corollary \ref{cor:2pos}, there are two possibilities. The first one is that there is an isometric cycle $C_{1} = (z_kz_{k+1}\ldots z_{l}w_{l-1}w_{l-2}\ldots w_{k+1})$, for some $0\leq k<l-1\leq n-1$ and some vertices $w_{l-1},w_{l-2},\ldots, w_{k+1}$ not on $P_1$. By Claim \ref{claim:1}, this is not possible. Hence the second possibility must hold: there exists a traverse $T'$ from $v_{2m+n_1}v_{2m+n_1-1}$ to $u_{2m+n_2}u_{2m+n_2-1}$, such that $P_1$ is the $v_{2m+n_1-1},u_{2m+n_2-1}$-side of $T'$. Let $P_0$ be the $v_{2m+n_1},u_{2m+n_2}$-side of $T'$.

Let $C'$, resp., $C''$ be the first isometric cycle on the traverse $T'$ from $v_{2m+n_1}v_{2m+n_1-1}$ to $u_{2m+n_2}u_{2m+n_2-1}$, resp.,  on the traverse $T$ from $v_{2m+n_1-1}v_{2m+n_1-2}$ to $u_{2m+n_2-1}u_{2m+n_2-2}$. Let $C'$ be of length $2k_1+2$, and $C''$ of length $2 k_2+2$. Without loss of generality, assume that $k_1\leq k_2$. Then $C'$ and $C''$ are isometric cycles having $k_1$ edges in common (which is at least 2, since $g(G)\geq 6$). Moreover, vertices of $C'$ that do not lie on $P_1$ are in $W_{v_{2m+n_1}v_{2m+n_1-1}}$, thus they do not intersect with $C''$. This implies that $C'$ and $C''$ intertwine. Using the fact that $k_2\leq n \leq n_1+n_2$ and $k_1\geq 2$, we get
$$i(C',C'')=(2k_1+2+2k_2+2-4k_1)/2=k_2-k_1+2\leq k_2 \leq n_1+n_2.$$
By the minimality condition, the above expression is an equality. This implies that $k_1=2$, i.e.~$C'$ is a 6-cycle. It also implies that $k_2=n_1+n_2$, i.e.~$C''$ is a $(2n_1+2n_2+2)$-cycle and thus $C''$ is the whole traverse $T$. Since the cycles $C'$ and $C''$ are again two isometric cycles that intertwine and have the minimal residue of intertwining, we can, without loss of generality, assume that the cycles $(v_0v_1\ldots v_mv_{m+1}\ldots v_{2m+2n_1-1})$ and $(u_0u_1\ldots u_mu_{m+1}\ldots u_{2m+2n_2-1})$ that we have started with are a 6-cycle and a $(2n_1+2n_2+2)$-cycle, respectively, i.e.~$n_1=1$ and $m=2$. 

Since the distance on $C''$ from $v_{2m+n_1-1} (=v_5)$ to $u_{2m+n_2-1}$ is $n_1+n_2=1+n_2$, we see that the length of $P_1$ is $1+n_2$ and thus also the length of $P_0$ and $P_2$ is $1+n_2$. The path $v_5v_0u_{2n_2+2m-1}\ldots u_{n_2+2m}$ has length $n_2+1$, thus it is a shortest path. If it is different from $P_0$, then, by Lemma \ref{lem:pastecycle}, we have an isometric cycle $E_1=(y_{i_1}y_{i_1+1}\ldots y_{i_2} y'_{i_2-1} \ldots y'_{i_1+1})$ for some vertices  $y_{i_1},\ldots,y_{i_2}$ on $P_0$ and some $y'_{i_1+1},\ldots,y'_{i_2-1}$ not on $P_0$. By Claim \ref{claim:1} this is not possible. Thus $P_0=v_5v_0u_{2n_2+2m-1}\ldots u_{n_2+2m}$.

Similarly,  the path $v_3v_2u_{3}\ldots u_{n_2+2}$ has length $1+n_2$, thus it is a shortest path. As above, if it is different from $P_2$, then Lemma \ref{lem:pastecycle} and Claim \ref{claim:1} give a contradiction. Thus $P_2=v_3v_2u_{3}\ldots u_{n_2+2}$. The cycle $(v_0v_1v_2v_3v_4v_5)$ is isometric, thus we have $v_0v_5 \Theta v_2v_3$. Since $C'$ and $C''$ are isometric, we have $v_0v_5 \Theta z_{1}z_{2}$ and $v_2v_3 \Theta z_{n_2}z_{n_2+1}$, thus $z_1z_2 \Theta z_{n_1}z_{n_1+1}$. But $z_1z_2$ and $z_{n_1}z_{n_2+1}$ lie on a shortest path $P_1$, thus $z_1z_2=z_{n_1}z_{n_1+1}$, i.e.~$n_2=1$. Therefore, also the cycle on $u_0,\ldots ,u_{2n_1+2m-1}$ is a 6-cycle. 

Consider the graph $H$ induced  on vertices $v_0,v_1,\ldots, v_{5},u_3,u_4,u_5, z_1$. We claim that $H$ is isomorphic to $X$. Since in $G$ we have the cycle $(v_0v_1v_2v_3v_4 v_{5})$, the cycle $(v_0v_1v_2u_3u_4 u_{5})$, and the path $P_1=v_3z_1u_3$, we see that $X$ is isomorphic to a spanning subgraph of $H$. But no additional edge can exist in $H$, since $g(G)\geq 6$. Finally we prove that $H$ is isometric. To prove this it is enough to check that for each pair of vertices $a,b\in V(H)$, there exists an $a,b$-path in $H$ that has all its edges in pairwise different $\Theta$-classes in $G$, i.e. it is a shortest path in $G$. If both $a,b$ lie in one of the isometric cycles $(v_0v_1\ldots v_{5})$, $(v_0v_1v_2u_3u_4 u_{5})$, $C'=(v_5v_4z_1u_4u_5v_0)$, or $C''=(v_3v_4z_1u_4u_3v_2)$, then this holds. The remaining pairs are $(v_5,u_3)$, $(u_5,v_3)$, and $(z_1,v_1)$ (see Figure 1). For the pair $(v_5,u_3)$, the path $v_5v_4v_3v_2u_3$ has its first three edges on a common convex cycle, so these edges are pairwise in different $\Theta$-classes. Moreover $v_2u_3$ is not in $F_{v_2v_3}$ since it is incident with $v_2v_3$ and not in $F_{v_3v_4}$ or $F_{v_4v_5}$ since it lies on a convex cycle with edges in these classes. For the remaining two pairs the situation is symmetric.
\end{proof}

Denote with
$$C(G)=\{ C \mid C \textrm{ is a convex cycle in } G\}.$$
The \emph{convex excess} of a graph $G$ was introduced in \cite{klavvzar2012convex} as
$$ce(G)=\sum_{C\in C(G)} \frac{|C|-4}{2},$$ 
and the following “Euler-type” formula was proved for partial cubes:
$$2n(G)-m(G)-i(G)-ce(G)\leq 2,$$
where $i(G)$ denotes the isometric dimension of $G$ (i.e., the number of $\Theta$-classes in $G$), $n(G)$ the number of vertices in $G$ and $m(G)$ the number of edges in $G$.
Moreover, the equality in the formula holds if and only if $G$ is a tree-zone graph. The  next result shows that there are many tree-zone partial cubes.

\begin{cor}
Every partial cube $G$ with $g(G)>6$ is a tree-zone graph and hence it holds $2n(G)-m(G)-i(G)-ce(G)= 2$.
\end{cor}
\begin{proof}
Let $uv\in E(G)$, and let $Z_{F_{uv}}$ be the $F_{uv}$-zone graph. Assume that we have a cycle in $Z_{F_{uv}}$. Then let $C_0,\ldots,C_{j}$ be a sequence of convex cycles for which $C_i$ and $C_{i+1}$ intersect in an edge from $F_{uv}$, where $i\in \{0,\ldots,j\}$, and $i+1$ is calculated in $\Z_{j+1}$. By Proposition \ref{prop:girth2}, pairs $C_i$ and $C_{i+1}$ intersect in exactly one edge. For  $i\in \{0,\ldots,j\}$, let $w_0^i,w_1^i,\ldots,w_{j_i}^i$ be vertices of $C_i$   that lie $W_{uv}$. Then $R=w_0^0w_1^0 \ldots w_{j_0}^0w_1^1w_2^1 \ldots w_{j_{j}}^{j}$ is a closed walk. Since two consecutive cycles $C_i,C_{i+1}$ share only an edge and that edge is from $F_{uv}$, we see that a sub-sequence (an interval) of $R$ forms a cycle. Let
$(u_1u_2\ldots u_k)$ be that cycle, for some $u_1,\ldots,u_{k_1}$ on $C_p$, $u_{k_1},\ldots,u_{k_1+k_2}$ on $C_{p+1}$, \ldots, and $u_{k-k_l},\ldots,u_{k}$ on $C_{p+l}$ for some $0\leq p<p+l\leq j$.

%

Let $u_{i_1}u_{i_1+1}\Theta  u_{i_2}u_{i_2+1}$ be  two edges on the cycle with $1\leq i_1<i_2<k$ and with $i_2-i_1$ as small as possible. The path $P=u_{i_1+1}\ldots u_{i_2}$ is a shortest path, since all the edges on $P$ are in pairwise different $\Theta$-classes. The latter holds: if two edges $u_{j_1}u_{j_1+1},u_{j_2}u_{j_2+1}$ on $P$ were in the same $\Theta$-class, we would have $j_2-j_1<i_2-i_1$.

 By Corollary \ref{cor:2pos}, either $P$ is the $u_{i_1+1}, u_{i_2}$-side of a traverse from $u_{i_1}u_{i_1+1}$ to $u_{i_2}u_{i_2+1}$, or there is an isometric cycle $D$ of the form $(u_{k_1},\ldots, u_{k_2},w_{k_2-1},\ldots,w_{k_1+1})$ for some $u_{k_1},\ldots, u_{k_2}$ on $P$. Since all the cycles $\{C_i; 0\leq i \leq j \}$, have an edge in $F_{uv}$, the cycle $D$ or the cycles of a traverse from $u_{i_1}u_{i_1+1}$ to $u_{i_2}u_{i_2+1}$ (whichever exists) are different from cycles $\{C_i; 0\leq i \leq j \}$.

Since $g(G)>6$, each isometric cycle $C_i$, $i\in \{0,1,\ldots,j\}$ has at least three consecutive edges on the closed walk $R$. If there exists the isometric cycle $D$, it has length at least 8, hence this cycle has at least 4 consecutive edges on $R$. Then it must share at least 2 edges with some $C_i$, $i\in \{0,1,\ldots,j\}$, which is a contradiction with Proposition \ref{prop:girth2}. On the other hand, if $P$ is a side of a traverse with isometric cycles of length at least 8, then each of this cycles has at least 3 consecutive edges on $R$ and it must share at least 2 edges with some $C_i$, $i\in \{0,1,\ldots,j\}$. A contradiction with Proposition \ref{prop:girth2}.

We have proven that no cycle exists in the $F_{uv}$-zone graph. Since $uv$ was arbitrary, the latter holds for all zone graphs of $G$.
\end{proof}

We notice that for the computation of $i(G)$ efficient algorithms have been developed, see \cite{cheng2012poset}.
To prove the main result of this paper, we will need the following:

\begin{lem}\label{lem:path}
Let  $g(G)>6$ for a partial cube $G$. If $u_1v_1 \Theta u_2v_2$ with $u_2 \in U_{u_1v_1}$, $P_1$ being a shortest $u_1u_2$-path, and $P_2$ being a shortest $v_1v_2$-path, then $P_1$ and $P_2$ are the sides of the unique traverse from $u_1v_1$ to $u_2v_2$.
\end{lem}

\begin{proof}
Let $P_1$ be a shortest $u_1u_2$-path, and let $R_1$ be the $u_1,u_2$-sides of some traverse $T$ from $u_1v_1$ to $u_2v_2$, provided by Lemma \ref{lem:cycles}. For the sake of contradiction, assume that $R_1\neq P_1$. By Lemma \ref{lem:pastecycle}, there exists an isometric cycle $C=(z_k\ldots z_{k+l} w_{k+l-1}\ldots, w_{k+1})$, where $z_k,\ldots, z_{k+l}$ are vertices on $R_1$ and $w_{k+l-1},\ldots, w_{k+1}$ are some other vertices. Since $g(G)>6$, the length of $C$ is at least 8, thous it has at least 4 consecutive edges on $R_1$. The length of the isometric cycles on $T$ is also at least 8, thus each has at least 3 consecutive edges on $R_1$. Hence there are two isometric cycles, namely $C$ and one of the isometric cycles on $T$, that have at least two edges in common. This is a contradiction with Proposition \ref{prop:girth2}.

We have proved that $R_1$ is the only shortest $u_1u_2$-path, and, similarly, the $v_1,v_2$-side of $T$ is the only shortest $v_1v_2$-path. Since it is impossible that two traverses have the same sides, this also proves the uniqueness of the traverse.
\end{proof}

We are now ready for our main result.
In the proof  we will use a rooted tree $T$ with root $v$. For every vertex $u\in V(T)$, we will denote the $v,u$-path in $T$ by $P_u$, and with $A_u$ the set of all the edges in $T$ that have exactly one endpoint in $V(P_u)\setminus\{u\}$.

\begin{thm}\label{thm:delta}
Every partial cube $G$ with $g(G)>6$ and $\delta(G)\geq 3$ contains an infinite  subtree in which vertices have degree 3 or 2. Moreover, any two vertices of degree 2 have distance at least 2. In particular, $G$ is infinite with exponential growth. 
\end{thm}

\begin{proof}
We will inductively build a claimed tree $T$. We will use a stronger induction hypothesis: We will assume that we have built a subtree $T_n$ such that all leafs have distance at least $n$ from the root, its vertices have degree at most 3, and any two vertices of degree 2 are at distance at least 2. Moreover, we will assume that vertices adjacent to leafs have degree 3, $v,u$-paths in $T$, for arbitrary $u\in V(T)$, are shortest paths in $G$, and for a fixed edge $wz \in E(T_n)$ the edges in $A_w$ are not in relation $\Theta$ with $wz$.
For the induction basis $T_1$ we can take an arbitrary root $v\in V(G)$ and three incident edges.

Now assume that we have built a subtree $T_n$ that satisfies the induction hypothesis. 
Pick any leaf $u$ of $T_n$, and let $u_{-1}$ be the neighbor of $u$ on $T_n$, and $P_u, A_u$ as defined before the theorem. Since  $\delta(G)\geq 3$, there are at least two neighbors of $u$ in $G$, distinct from  $u_{-1}$. Denote them with $u_1,u_2$. Assume that none of the edges $uu_1, uu_2$ is in relation $\Theta$ with an edge on $P_u$ or an edge in $A_u$. Then we extend $T_n$ with $uu_1$ and $uu_2$. Let us prove that we obtain a tree that satisfies the induction hypothesis. 

Since $uu_1$ and $uu_2$ are not in relation $\Theta$ with any edge on $P_u$, the $v,u_1$- and $v,u_2$-path in the tree are shortest paths in $G$. We have to check that $u_1$ or $u_2$ are not vertices of $T_n$, since in this case we would have obtained a cycle by adding edges $uu_1$ and $uu_2$. If $u_1$ is already on $T_n$, then denote with $C$ the obtained cycle. Let $ab$ be the edge on $C$ that is in $A_u$. By the definition of $A_u$, $ab\neq uu_1$ and $ab$ has exactly one endpoint on $P_u$, say $a$ is on $P_u$. Then there is at least one another edge on $C$ which is in $F_{ab}$. By induction assumption, all the edges on the $a,u_1$-path on $T_n$ (the path in the non-extended tree) are in different $\Theta$ classes since this path is a shortest path in $G$. Moreover, all the edges on the $a,u$-path in $T_n$ are not in relation $\Theta$ with edges of $A_u$, by induction assumption, in particular, none of them is in relation $\Theta$ with $ab$. Also, $uu_1$ is not in relation $\Theta$ with $ab$, by our assumption. A contradiction. Similarly, we prove that  $u_2\notin V(T_n)$. All the other induction assumptions are trivially satisfied. We have proved, that in this case we can extend $T_n$ with edges $uu_1$ and $uu_2$.

Now assume that $uu_1$ is in relation $\Theta$ with an edge $ab$ on $P_u$ or in $A_u$ (with $a$ closer to $u$ than $b$). In both cases, by Lemma  \ref{lem:path}, the $a,u$-path on $T_n$ is a side of the traverse from $uu_1$ to $ab$. The letter implies that $u$ is at distance at least 3 from the root $v$, since $g(G)\geq 8$.  Let $u_{-2}$ and $u_{-3}$ be the third last and forth last vertices on $P_u$, respectively. Since the girth of $G$ is at least 8, the path $uu_{-1}u_{-2}u_{-3}$ lies on an isometric cycle $C'$, the first isometric cycle of the traverse from $uu_1$ to $ab$. If also $uu_2$ is in relation $\Theta$ with an edge on $P_u$ or in $A_u$, the path $uu_{-1}u_{-2}u_{-3}$ would lie on another isometric cycle, which is a contradiction with Proposition \ref{prop:girth2}. Thus we can extend $T_n$ with $uu_2$, and obtain a subtree $T_{n}'$, which satisfies all the induction assumptions, apart from the assumption that vertices adjacent to leafs have degree 3.

We can extend $T_{n}'$ a bit more. Denote with $u_3,u_4$ two neighbors of $u_2$ in $G$ distinct from $u$. If none of the edges $u_2u_3, u_2u_4$ is in relation $\Theta$ with an edge on $P_{u_2}$ or an edge in $A_{u_2}$, we can extend $T'_{n}$ with both of them to obtain a subtree that satisfies the induction hypothesis, by the same arguments as before. On the other hand, if $uu_3$ is in relation $\Theta$ with an edge on $P_{u_2}$ or an edge in $A_{u_2}$, then path $u_{2}uu_{-1}u_{-2}$ lies on an isometric cycle. This cycle is clearly distinct from $C'$, but they share more than an edge. A contradiction with Proposition \ref{prop:girth2}. Thus, we can always extend $T_{n}'$.

We can extend in this way all the leafs in $T_n$ with distance less than $n+1$ from the root and obtain a tree $T_{n+1}$. By induction, an infinite tree from the theorem exists. The last assertion of the theorem now easily follows.
\end{proof}

\begin{cor}\label{cor:regulargirth}
Let $G$ be a finite regular partial cube with $g(G) > 6$. Then $G$ is $K_1$, $K_2$ or an even cycle.
\end{cor}

To see that the condition $g(G)>6$ in Theorem \ref{thm:delta} and Corollary \ref{cor:regulargirth} cannot be weakened, consider the following example.
Recall that the \emph{middle level} graph $M_{2n+1}$, for $n\geq 1$, is the subgraph of $Q_{2n+1}$ induced on the vertices $(i_1,\ldots,i_{2n+1})$, such that there are exactly $n$ or $n+1$ coordinates equal to 1.
In particular, $M_3$ is the cycle of length 6, while $M_5$ is known as the Desargues graph. Middle level graphs are the only distance-regular partial cubes with girth 6 \cite{weichsel1992distance}, and they show that the bound $g(G)>6$ is tight. Notice that in the case $n\geq 2$ these graphs have many isometric subgraphs isomorphic to $X$. 

One could consider partial cubes with $\delta(G)\geq 3$, $g(G)=6$, and no isometric subgraphs isomorphic to $X$. One example of such a graph is an infinite hexagonal net. It is clearly infinite with non-exponential (polynomial) growth. We know of no finite example of such a graph.


Finally, in view of Corollary \ref{cor:regulargirth}, notice that is quite easy to construct regular partial cubes of higher degrees with girth 4. If we take the Cartesian product of any two regular partial cubes, we get a regular partial cube of girth 4. Simple examples are $Q_n \B C_{2m}$, for every $n\geq 1,m\geq 2$, where $Q_n$ is a hypercube of dimension $n$. For more cubic graphs that can be used as factors, see \cite{klavvzar2007tribes}.

%

\section*{Acknowledgment}
The author wishes to express his gratitude to Sandi Klavžar who suggested to study the topic and gave useful comments on the text.

\section*{References}
\bibliographystyle{plain}
\bibliography{biblio}

\end{document}